\documentclass[10pt, conference, letterpaper]{IEEEtran}
\usepackage[utf8]{inputenc}
\usepackage[english]{babel}
\usepackage{caption}
\usepackage{authblk}
\usepackage{graphicx}

\usepackage{xspace}
\usepackage{graphics} 
\usepackage{mathtools, bm}

\usepackage{amssymb, bm}
\usepackage{complexity}
\usepackage{amsmath,amsthm}
\usepackage{pgfplots}
\usepackage{url}

\usepackage{color}
\usepackage{amsmath}
\usepackage[colorlinks=true,breaklinks=true,linkcolor=blue]{hyperref}

\newtheorem{prop}{Proposition}

\newtheorem{corollary}{Corollary}

\graphicspath{{img/}}

\title{Deterministic contention management for low latency Cloud RAN over an optical ring}
\makeatletter
\renewcommand*{\thetable}{\arabic{table}}
\renewcommand*{\thefigure}{\arabic{figure}}
\let\c@table\c@figure
\makeatother 

\author[1]{Dominique Barth}
\author[1,2]{Ma\"el Guiraud}
\author[1]{Yann Strozecki}
\affil[1]{David Laboratory, UVSQ}
\affil[2]{Nokia Bell Labs France}
\begin{document}
\maketitle

\begin{abstract}
The N-GREEN project has for goal the design of a low cost optical ring technology with good performance (throughput, latency$\dots$) without using expensive end-to-end connections. We study the compatibility of such a technology with the development of the Cloud RAN, a latency critical application which is a major aspect of 5G deployment. We show that deterministically managing Cloud RAN traffic minimizes its latency while also improving the latency of the other traffics. 

\end{abstract}

\section{Introduction}

\footnote{This work was developed for the N-GREEN project. The authors thank the National Agency of Research (ANR) for partial funding in the frame of the N-GREEN project and the parteners of the project for fruitful discussions.} Telecommunication network providers have to design inexpensive networks supporting an increasing amount of data and online applications. Many of these applications require QoS guarantees, like minimal throughput and/or  maximal latency. The N-GREEN project aims to design a high performing optical ring while ensuring a minimal cost for providers. The current solutions with good QoS~\cite{pizzinat2015things,tayq2017real}, establish end-to-end direct connections between the nodes, which is extremely expensive. The N-GREEN optical ring, offering any-to-any connections, is designed to ensure good performance at low cost: beyond the advantages of WDM technology adopted, the hardware it requires scales linearly with the number of nodes while direct connection scales quadratically making it impractical for more than a few nodes. The WDM technology of the N-GREEN optical ring is different of existing technologies or protocols like SDH/SONET and DQBD~\cite{siller1996sonet,zukerman1990dqdb}.

In this article, we study a Cloud RAN (C-RAN) application based on the N-GREEN optical ring described in~\cite{ngreenarchitecture,uscumlic2018scalable}. C-RAN is one of the major area of development for 5G; it consists in centralizing or partially centralizing the computation units or {\bf BaseBand Units} (BBU) of the {\bf Remote Radio Heads} (RRH) in one datacenter~\cite{mobile2011c}. Periodically, each RRH in the field sends some uplink traffic to its associated BBU in the datacenter, then, after a computation, the BBU sends some downlink traffic back to the RRH. In this paper, we assume that the quantity of uplink and downlink traffic is the same. The latency of the messages between the BBU and the RRH is critical since some services need end-to-end latency as low as $1$ms~\cite{3gpp5g,boccardi2014five}.

Nowadays, the traffic is managed by statistical multiplexing~\cite{kern2006applying}. Here, we propose an SDN approach to {\bf deterministically} manage the periodic C-RAN traffic by choosing emission timing. Indeed, Deterministic Networking is one of the main method considered to reduce the end-to-end latency~\cite{finn-detnet-architecture-08}. In a previous work~\cite{dominique2018deterministic}, the authors have studied a similar problem for a star shaped network. In contrast with our previous work, finding emission timings so that different periodic sources do not use the same resource is easy in the context of the N-GREEN optical ring with a single data-center. However, we deal with two additional difficulties arising from practice: the messages from RRHs are scattered because of the electronic to optic interface and there are other traffics whose latency must be preserved. It turns out that the deterministic management of CRAN traffic we propose reduces the latency of CRAN traffic to the physical delay of the routes, while reducing the latency of the other traffics by smoothing the load of the ring over the period. To achieve such a good latency, our solution needs to reserve resources in advance, which slightly decreases the maximal load the N-GREEN optical ring can handle. Such an approach of reservation of the network for an application (CRAN in our context) relates to network slicing~\cite{jiang2016network} or virtual-circuit-switched connections in optical networks~\cite{cadere2010virtual,szymanski2016ultra}.

In Sec.~\ref{sec:model}, we model the optical ring and the traffic flow. In Sec.~\ref{sec:oportmethods}, we experimentally evaluate the latency when using stochastic multiplexing to manage packets insertion on the ring, with or without priority for C-RAN packets. In Sec.~\ref{sec:deterministicalgorithms}, we propose a deterministic way to manage C-RAN packets without buffers, which guarantees to have zero additional latency from buffering in the optical ring. We propose several refinements of this deterministic sending scheme to spread the load over time, which improves the latency of best effort packet, or in Sec.~\ref{sec:maxant}, to allow the ring to support a maximal number of antennas at the cost of a very small latency for the C-RAN traffic. 

\section{Model of C-RAN traffic over an optical ring}
\label{sec:model}
    
  \paragraph{N-GREEN Optical ring}
   
  The unidirectional optical ring is represented by an oriented cycle. The vertices of the cycle represent the nodes of the ring, where the traffic arrives. The arcs $(u,v)$ of the cycle have an integer weight $\omega(u,v)$ which represents the time to transmit a unit of information from $u$ to $v$. By extension, if $u$ and $v$ are not adjacent, we denote by $\omega(u,v)$ the size of the directed path from $u$ to $v$.  The \textbf{ring size} is the length of the cycle, that is $\omega(u,u)$ and we denote it by $RS$. A {\bf container}, of capacity $C$  expressed in bytes, is a basic unit of data in the optical ring. 
  
  The time is discretized: a unit of time corresponds to the time needed to fill a container with data.
  As shown in Fig.~\ref{fig:containers}, the node $u$ can fill a container with a data packet of size less than $C$ bytes at time $t$ if the container at position $u$ at time $t$ is \emph{free}. 
  If there are several packets in a node or if a node cannot fill a container, because it is not free, 
  the remaining packets are stored in the {\bf insertion buffer} of the node. 
  A container goes from $u$ to $v$ in $\omega(u,v)$ units of time. The ring follows a {\bf broadcast and select scheme with emission release policy}: When a container is filled by some node $u$, it is freed when it comes back at $u$ after going through the whole cycle.

  \begin{figure}[h!]

        \begin{center}
   \includegraphics[scale=0.65]{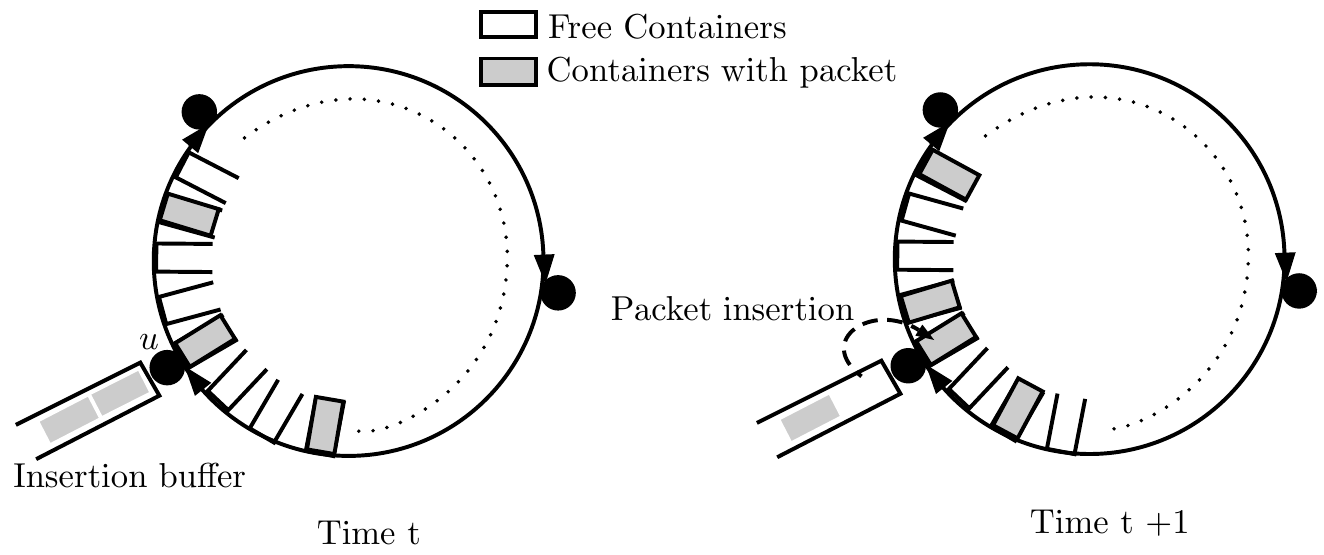}
   \end{center} 
     \caption{Dynamic behavior of the ring.}\label{fig:containers}
  \end{figure}

     \paragraph{C-RAN traffic}
   The RRHs are the source of the {\bf deterministic and periodic} C-RAN traffic.
   There are $k$ RRHs attached to the ring and several RRHs can be attached to the same vertex. An RRH is linked to a node of the ring through an electronic interface of bit rate $R$ Bps.
   The ring has a larger bit rate of $F\times R$ Bps. The integer $F$ is called the {\bf acceleration factor} between the electronic and the optical domains. A node aggregates the data received on the electronic interface during $F$ units of time to create a packet of size $C$ and then puts it in the insertion buffer.
%   fills a container of the ring with this data. 
  In each period $P$, an RRH emits data during a time called \textbf{emission time} or $ET$. Hence the RRH emits $ET / F$ packets, i.e. requires a container of size $C$ each $F$ units of time during the emission time, as shown in Fig.~\ref{fig:interface}.
   %Each period $P$, an RRH emits $ET / F$ packets, i.e. a packet of size $C$ each $F$ units of time during a time $ET$ (emission time),
   
   At each period, the data of the RRH $i$ begins to arrive in the insertion buffer at a time $m_i$  called {\bf offset}. The offsets can be determined by the designer of the system and can be different for each RRH but must remain the same over all periods. We assume that all BBUs are contained in the same data-center attached to the node $v$. The data from $u$ is routed to its BBU at node $v$ through the ring and arrives at time $m_i + \omega(u,v)$ if it has been inserted in the ring upon arrival. Then after some computation time, which w.l.o.g. is supposed to be zero, an answer is sent back from the BBU to the RRH. The same quantity of data is emitted by each BBU or RRH during any period.
   
  In this paper, the {\bf latency} of a data packet is defined as the time it waits in an insertion buffer. In other words, it is the logical latency into the optical ring.
   Indeed, because of the ring topology, the routes between RRHs and BBUs are fixed, thus we cannot reduce the physical transmission delay of a data which depends only on the size of the arcs used. Moreover, there is only one buffering point in the N-GREEN optical ring, the insertion buffer of the node at which the data arrives. Hence, in this context, to minimize the end-to-end delay, we need to minimize the (logical) latency.
   More precisely, we want to reduce the latency of the C-RAN traffic to \textbf{zero}, both for the RRHs (uplink) and the BBUs (downlink). In Sec.~\ref{sec:deterministicalgorithms} we propose a deterministic mechanism with zero latency for C-RAN which also improves the latency of other data going through the optical ring. We shortly describe the nature of this additional traffic in the next paragraph.
     
\begin{figure}[h!]
\begin{center}  
      \includegraphics[scale=0.7]{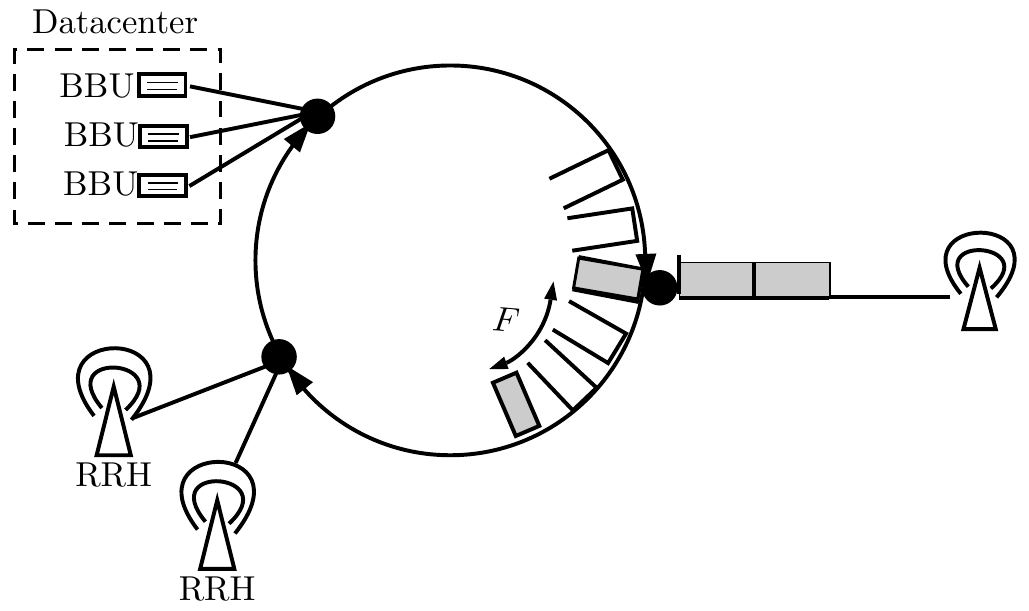}
     \caption{Insertion of C-RAN traffic in the N-GREEN optical ring.}\label{fig:interface}
\end{center}
  \end{figure}

\paragraph{Best effort traffic}
The optical ring supports other traffics, corresponding to the internet flow. We call this traffic \textbf{Best Effort} (BE). We want it to have the best possible distribution of latency, but since BE traffic is less critical than C-RAN traffic, we impose no hard constraint on its latency. At each node of the ring, a {\bf contention buffer} is filled by a batch arrival process of BE data. 
This batch arrival process consists in generating, at each unit of time, a quantity of data drawn from a bimodal distribution to modelize the fact that internet traffic is bursty. Then, according to the fill rate of the contention buffer and the maximum waiting time of the data, a packet of size at most $C$ may be created by aggregating data in the contention buffer. This packet is then put in the insertion buffer of the node. Hence, the arrival of BE messages can be modeled by a temporal law that gives the distribution of times between two arrivals of a BE packet in the insertion buffer. The computation of this distribution for the parameters of the contention buffer used in the N-GREEN optical ring is described in~\cite{Cast1810:Performance}. We use this distribution in our experiments to modelize BE packet arrival in the insertion buffer.

   \section{Evaluation of the latency on the N-GREEN optical ring}
   \label{sec:oportmethods}

  We first study the latency of the C-RAN and BE traffics when the ring follows an opportunistic insertion policy: When a free container goes through a node, it is filled with a packet of its insertion buffer, if there is one.
 Two different methods to manage the insertion buffer are experimentally compared. First, the \textbf{FIFO} rule, which consists in managing the C-RAN and BE packets in the same insertion buffer. Then, when a free container is available, the node fills it with the oldest packet of the insertion buffer, without distinction between C-RAN and BE. This method is compared to a method called \textbf{C-RAN priority} that uses two insertion buffers: one for the BE packets, and another for the C-RAN packets. The C-RAN insertion buffer has the priority and is used to fill containers on the ring while it is non empty before considering the BE insertion buffer.  
 
We compare experimentally these two methods in the simplest topology: The lengths of the arcs between nodes are equal and there is one RRH by node. The experimental parameters are given in Table~\ref{fig:params} and chosen following~\cite{ngreenarchitecture}. In each experiment, the offsets of the RRHs are drawn uniformly at random in the period. The results are computed over $1,000$ experiments in which the optical ring is simulated during $1,000,000$ units of time. Fig.~\ref{fig:resultopport} gives the cumulative distribution of both C-RAN and BE traffics latencies for the FIFO and the C-RAN priority methods. The source code in C of the experiments can be found on one of the authors' webpage~\cite{webpage}.

\begin{center}
 \includegraphics[scale=0.35]{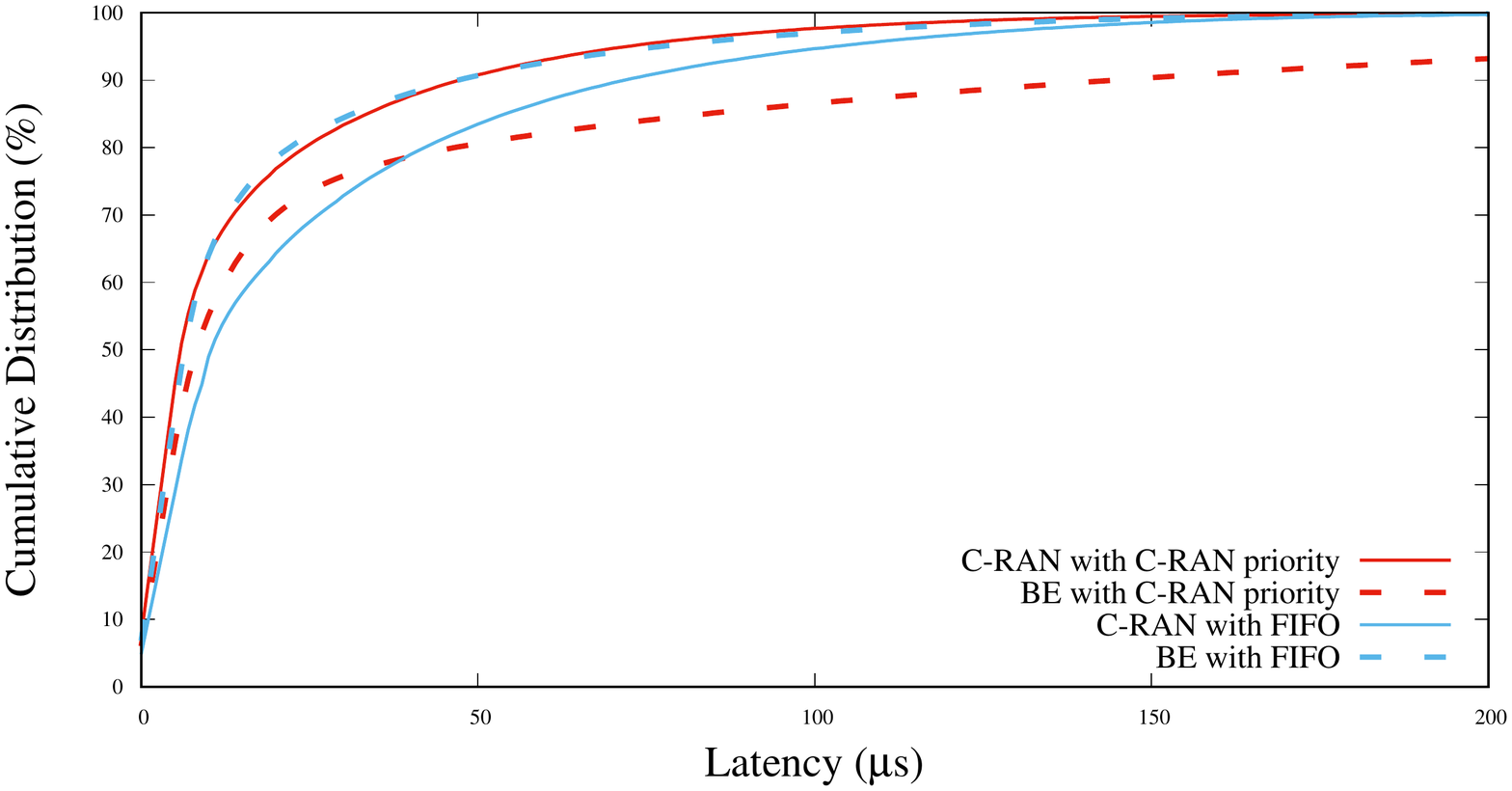}

     \captionof{figure}{Distribution of latencies for FIFO and C-RAN first.}    \label{fig:resultopport}

\begin{minipage}[c]{.45\linewidth}

 \scalebox{.57}
  {
  \begin{tabular}{|c|c|}
  \hline
  Bit rate of an electronic interface $R$ & $10$ Gbps \tabularnewline
  \hline
  Optical ring bit rate $F\times R$ & $100$ Gbps \tabularnewline
  \hline
    Acceleration factor $F$ & $10$  \tabularnewline
  \hline
  Container size  $C$ & $100$ kb  \tabularnewline
  \hline
  Unit of time (UoT) $C/(F\times R)$ & $1~\mu$s \tabularnewline
  \hline
  Length traveled during one UoT & $200$ m \tabularnewline
  \hline
    \end{tabular}
  }
  \end{minipage} % ne pas sauter de ligne
  \hspace{0.1cm}
  \begin{minipage}[c]{.45\linewidth}
   \scalebox{.57}
  {
  \begin{tabular}{|c|c|}
  \hline
  Time to go through the cycle $RS$ & $100$ UoT \tabularnewline
  \hline
  Emission time $ET$ & $500$ UoT \tabularnewline
  \hline
   Period $P$ & $1,000$ UoT \tabularnewline
  \hline
  Number of RRH & $5$  \tabularnewline
  \hline
  Number of nodes $k$ & $5$  \tabularnewline
  \hline
   Load induced by C-RAN traffic & $50\%$  \tabularnewline
  \hline
    Load induced by BE traffic & $40\%$  \tabularnewline
  \hline
  \end{tabular}
}
 
\end{minipage}
\captionof{table}{Parameters of the N-GREEN architecture.}\label{fig:params}
\end{center}
Unsurprisingly, the latency of the C-RAN traffic is better when we prioritize the C-RAN messages, while the BE traffic is heavily penalized. Furthermore, there is still $10\%$ of the C-RAN traffic with a latency higher than $50 \mu$s, a problem we address in the next section.

% Remark that the C-RAN and BE traffic latency distribution is not the same with the FIFO method. This comes from the fact that these two sources of traffic are of different nature. With high probability, there are some part of the period which have more 
% BE data arrivals (or C-RAN traffic) than average. In both cases, it makes the latency higher during these times. Furthermore, C-RAN traffic is concentrated on half the period, while BE traffic is distributed in all the period.\todo{revoir cette explication}
% % However a C-RAN emission from an RRH takes half the period while a BE arrival corresponds to only a few unit of times,
% It explains why the C-RAN traffic suffers more from this phenomena.

Remark that, due to the broadcast and select mode, a message coming from any node induces the same load for all the nodes of the ring. Hence the latency of the traffics coming from any RRHs or from the BBUs are the same, which may seem couterintuitive knowing that all BBUs share the same node on the ring. This is why in Fig.~\ref{fig:resultopport} we do not ditinguish between uplink C-RAN traffic (RRH to BBU) and downlink  C-RAN traffic (BBU to RRH).

\section{Deterministic approach for zero latency} \label{sec:deterministicalgorithms}
\subsection{Reservation}
Finding good offsets for the C-RAN traffic is a hard problem even for simple topologies and without BE traffic, see~\cite{dominique2018deterministic}. In this section, we give a simple solution to this problem in the N-GREEN optical ring, and we adapt it to minimize the latency of the BE traffic.

Let $u$ be the node to which is attached the RRH $i$. To ensure zero latency for the C-RAN traffic, then the container which arrives at $u$ at time $m_i$ must be free so that the data from the RRH can be sent immediately on the optical ring. 

To avoid latency between the arrival of the data from the RRH and its insertion on the optical ring, 
we allow nodes to \textbf{reserve} a container one round before using it. A container which is reserved cannot be filled by any node except the one which has reserved it (but it may not be free when it is reserved). 
If $u$ reserves a container at time $m_i - RS$, then it is guaranteed that $u$ can fill a free container at time $m_i$ with the data of the RRH $i$.
In the method we now describe, the C-RAN packets never wait in the node: The message sent by the RRH $i$ arrives at its BBU at node $v$ at time $m_i + \omega(u,v)$ and the answer is sent from the BBU at time $m_i + \omega(u,v) +1$.

Recall that an RRH fills a container every $F$ units of time, during a time $ET$. 
Thus if we divide the period $P$ into \textbf{slots} of $F$ consecutive units of time, an RRH needs to fill at most one container each slot. If an RRH emits at time $m_i$, then we say it is at \textbf{position} $m_i + \omega(u,v)\pmod F$.
The position of an RRH corresponds to the position in a slot of the container it has emitted, when it arrives at $v$, the node of the BBU. 
If an RRH is at position $p$, then by construction, the corresponding $BBU$ is at position $p+1\pmod F$. For now, we do not allow waiting times for C-RAN traffic, hence each RRH uses a container at \emph{the same position during all the emission time}. 

Given a ring, a set of RRH's, a period and an acceleration factor $F$, the problem we solve here is to find an \textbf{assignment} of values of the offsets $m_i$'s which is \textbf{valid}: two RRHs must never use the same container in a period. Moreover we want to preserve the latency of the BE traffic. It means that the time a BE packet waits in the insertion buffer must be minimized. To do so, we must minimize the time a node waits for a free container at any point in the period, by spreading the C-RAN traffic as uniformly as possible over the period. % in order to give the nodes an available container in a minimal average time. 

\begin{center}   
      \includegraphics[scale=0.65]{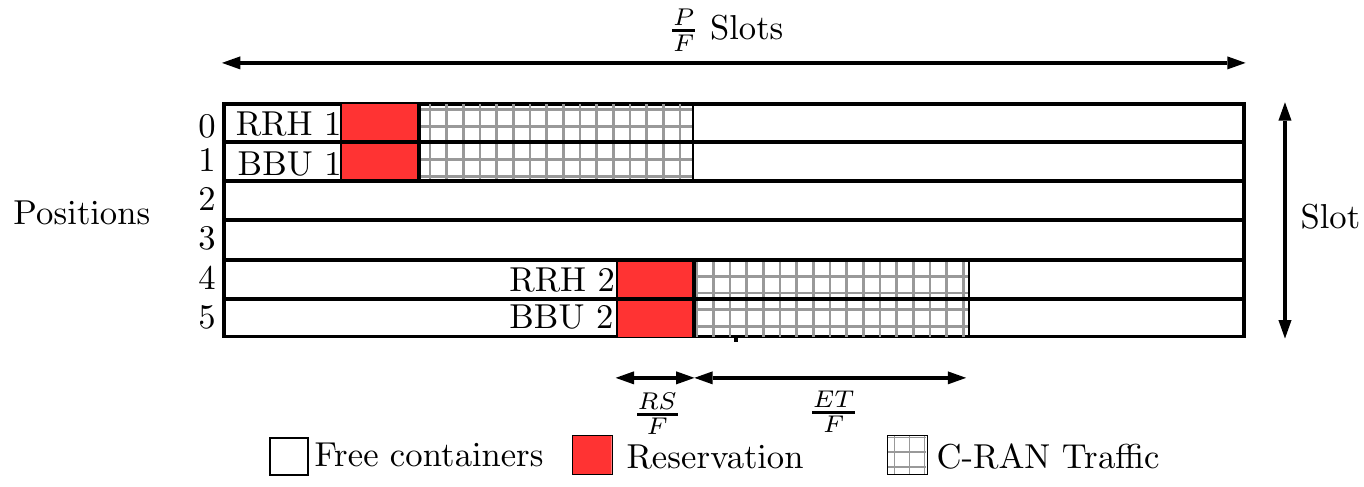}
     \captionof{figure}{A valid assignment with $F = 6$.}\label{fig:assignment}
\end{center}

Fig.~\ref{fig:assignment} represents an assignment of two couples of RRH and BBU by showing the containers going through the node of the BBU during a period. Each slot has a duration of $F$ unit of times, and, since an RRH/BBU emits a packet each $F$ UoT during $ET$ UoT, if we take the granularity of a slot to represent the time, the emission of a BBU/RRH is continuous in our representation, during $ET/F$ slots. A date $t$ in the period corresponds in Fig.~\ref{fig:assignment} to the slot $t/F$ and is at position $t \mod F$.

 \subsection{Building valid assignment with zero C-RAN latency}\label{sec:zerolatency}
Remark that two RRHs which are not at the same position never use the same containers. Moreover, if we fix the offsets of the RRHs to even positions so that they do not reserve the same containers, then, because the answers of the BBU are sent without delay in our model, it will fix the offsets of the BBUs to odd positions which do not reserve the same containers. Hence, we need to deal with the RRHs only.
The next proposition gives a simple method to find an assignment.

\begin{prop}
\label{prop:assign}
There is a valid assignment of the offsets $m_1, \dots, m_k$ on the same position if  $k ET + RS \leq P$.
\end{prop}
\begin{proof}
 W.l.o.g we fix $m_1$ to $0$ and all the other offsets will then be chosen at position $0$.  Let $u_1,\dots,u_k$ be the nodes attached to the RRHs $1,\dots,k$. We assume that $u_1,\dots,u_k$ are in the order of the oriented cycle. The last message emitted by the RRH $1$ arrives at $u_2$ at time $ET - 1 + \omega(u_1,u_2)$. Therefore we can fix $m_2 =  ET  + \omega(u_1,u_2)$. In general we can set $m_i = (i-1) \times ET + \omega(u_1,u_i)$ and all RRHs will use different containers at position $0$ during a period. Since $k \times ET + \omega(u_1,u_1) \leq P$ by hypothesis,
 the containers filled by the $k$th RRH are freed before $P$. Hence when the RRH $1$ must emit something at the first unit of time of the second period, there is a free container.
\end{proof}

Remark that reserving free containers make them unusable for BE traffic which is akin to a loss of bandwidth. However, with our choice of emission times of the RRHs in the order of the cycle, most of the container we reserve are used by the data from some RRH. If all containers at some position are used, that is $kET +RS = P$, then there are only $RS$ free containers wasted. In the worst case, less than $2RS$ containers are wasted by the assignment of Prop.~\ref{prop:assign}. 

It is now easy to derive the maximal number of antennas which can be supported by an optical ring, when using reservation and the same position for an RRH for the whole period.

\begin{corollary}
There is a valid assignment with $ \lfloor\frac{P- RS}{ET}\rfloor \times \frac{F}{2}$ antennas and zero latency.
\end{corollary}
\begin{proof}
Following Prop.~\ref{prop:assign}, the maximal number of antennas for which there is an assignment on the same position is $k = \lfloor\frac{P- RS}{ET}\rfloor $.
In such an assignment, we need a second position to deal with the traffic coming from the BBUs coming back to those k antennas. Since we got  $F$ positions in the slot, the number of antennas supported by the ring is thus equal to $k \times \frac{F}{2}$.
\end{proof}

With the parameters of the N-GREEN ring given in Table.~\ref{fig:params}, we can support $5$ antennas, while stochastic multiplexing can support $10$ antennas albeit with extreme latency. There are two sources of inefficiency in our method. The first comes from the reservation and cannot be avoided to guarantee the latency of the C-RAN traffic. The second comes from the fact that an RRH must emit at the same position during all the emission time (to guarantee zero latency). We relax this constraint in Sec.~\ref{sec:maxant} to maximize the number of antennas supported by the ring, while minimizing the loss of bandwith due to reservation.

We now present an algorithm using reservation as in Prop.\ref{prop:assign} to set the offsets of several RRHs at the same position. In a naive assignment, we put each RRH in an arbitrary position, for instance one RRH by position. We then propose three ideas to optimize the latency of the BE traffic, by spacing as well as possible the free containers in a period.

\paragraph{Balancing inside the period}

With the parameters of the N-GREEN ring given in Table.~\ref{fig:params} ($ET = \frac{P}{2}$, $F = 10$ and $n = 5$), there are no unused position. Any assignment has exactly one  BBU or RRH at each position. If all the RRHs start to emit at the first slot, then during $ET$ there will be no free container anywhere on the ring, inducing a huge latency for the BE traffic. 
To mitigate this problem, in a period, the time with free containers in each position must be uniformly distributed over the period as shown in Fig.~\ref{fig:periodbal}.

\begin{center}
      \includegraphics[scale=0.6]{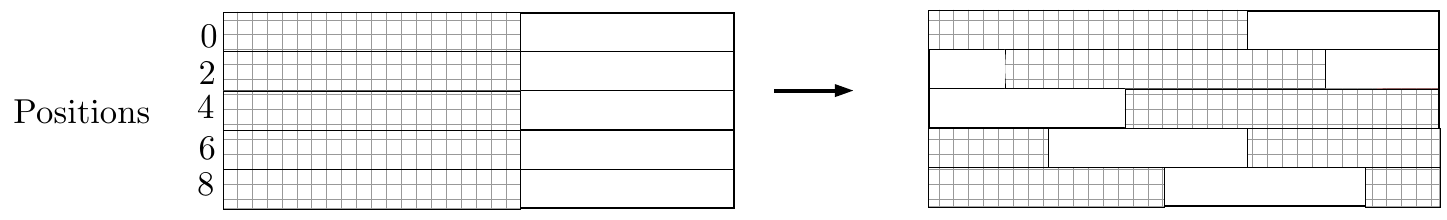}
     \captionof{figure}{Balancing inside the period.}\label{fig:periodbal}
\end{center}

\paragraph{Compacting positions}

For each position which is used by some RRH, and for each period, at least $RS$ free containers are reserved which decreases the maximal load the system can handle. Therefore to not waste bandwidth, it is important to put as many RRHs as possible on the same position as shown in Fig.~\ref{fig:packing}. Indeed, for any position which is not used at all, no container needs to be reserved. This strategy is also good to spread the load during the period since it maximizes the number of unused positions and for each unused position there is a container free of C-RAN traffic each $F$ unit of times. 
    
\begin{center}
      \includegraphics[scale=0.6]{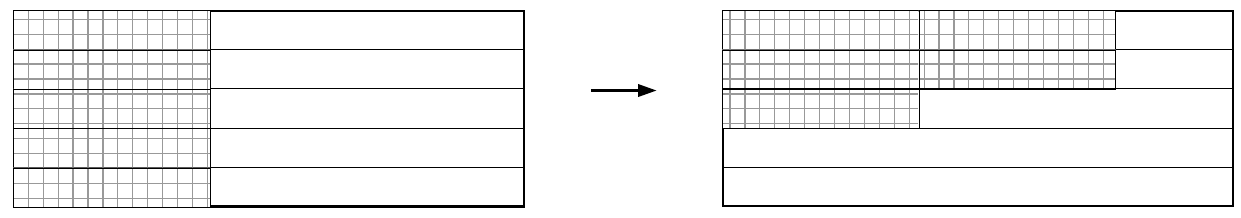}
     \captionof{figure}{Compacting positions.}\label{fig:packing}
\end{center}

\paragraph{Balancing used positions}

The free positions can be distributed uniformly over a slot, to minimize the time to wait before a node has access to a container from a free position, as shown in Fig.~\ref{fig:slotbal}. To do so, compute the number of needed positions $x = \lceil k\times \frac{ET}{P - RS}\rceil$, with k the number of antennas using the previous strategy. Then, set the $x$ used positions in the following way: $\lfloor\frac{F}{x}\rfloor -1 $ free positions are set between each used positions. If $\frac{F}{x}$ has a reminder $r$, then we set the $r$ free remaining positions uniformly over the interval in the same way and so on until there are no more free position. It is a small optimization, since it decreases the latency by at most $F/2$.
\begin{center}
   \includegraphics[scale=0.6]{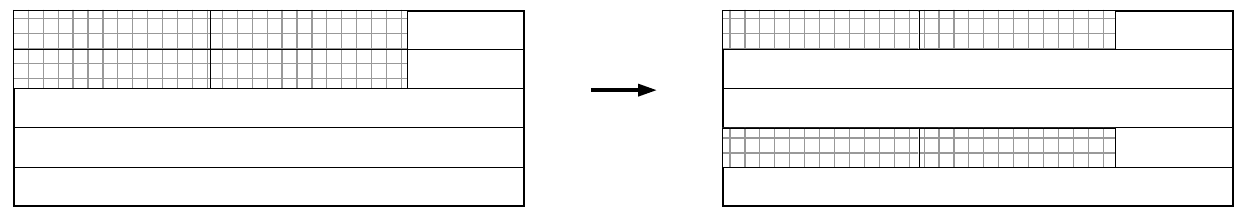}
     \captionof{figure}{Balancing used positions.}\label{fig:slotbal}
     \end{center}
  \paragraph{Experimental evaluation}

  Our algorithm \emph{combines the three methods} we have described to spread the load over the period.
  In order to understand the interest of each improvement, we present the cumulative distribution of the latency of the BE traffic using them either alone or in conjunction and we compare our algorithm to stochastic multiplexing with C-RAN priority.

  \begin{center}
    \includegraphics[scale=0.3]{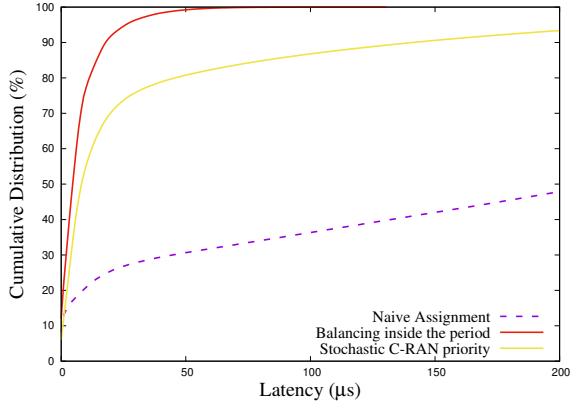}
     \captionof{figure}{BE latencies of a naive assignment and balancing inside the period for $5$ antennas.}  
     \label{fig:periodonly}
     \end{center}

Fig.~\ref{fig:periodonly} shows the performance of balancing the C-RAN traffic inside the period against a naive assignment in which all the RRH begin to emit at the same slot. We keep the same parameters as in Sec.~\ref{sec:oportmethods} (see Table~\ref{fig:params}). As expected, the BE traffic latency is much better when we balance the C-RAN traffic inside the period and already much better than stochastic multiplexing.

To show the interest of compacting the positions, we must be able to put several RRHs at the same position.
Hence, we change the emission time to $ET = 200$ and the number of antennas to $k = 12$ to keep the load around $90\%$ as in the experiment of Fig.~\ref{fig:resultopport}. This is not out of context since the exact split of the C-RAN (the degree of centralization of the computation units in the cloud) is not fully determined yet~\cite{mobile2011c}.

%With these parameters, the loss of bandwidth due to reservation is at most $6\%$.

As shown in Fig.\ref{fig:algocmp}, the performance of the naive assignment is really bad. Compacting the RRHs on a minimal number of positions decreases dramatically the latency. If in addition, we balance over a period, we get another gain of latency of smaller magnitude: the average (respectively maximum) latency for BE traffic goes from $4.76 \mu$s (resp. $48 \mu$s) to $3.28 \mu$s (resp. $37 \mu$s).
We did not represent the benefit of balancing used positions because the reduction in latency it yields is small as expected: the average (respectively maximum) latency for BE traffic goes from $4.76 \mu$s (resp. $48 \mu$s) to $4.43 \mu$s (resp. $44 \mu$s).

In Fig.~\ref{fig:optimres}, we compare the cumulative distribution of the latency of the BE traffic using the FIFO rule to our reservation algorithm with the three proposed improvements. The parameter are the same as in the previous experiment. The performance of our reservation algorithm is excellent, since the C-RAN traffic has \emph{zero latency} and the BE traffic has a \emph{better latency} than with the FIFO rule despite the cost of reservation. It is due to the balancing of the load of the C-RAN traffic over the period, that guarantee a more regular bandwidth for the BE traffic.

\begin{center}
 \includegraphics[scale=0.3]{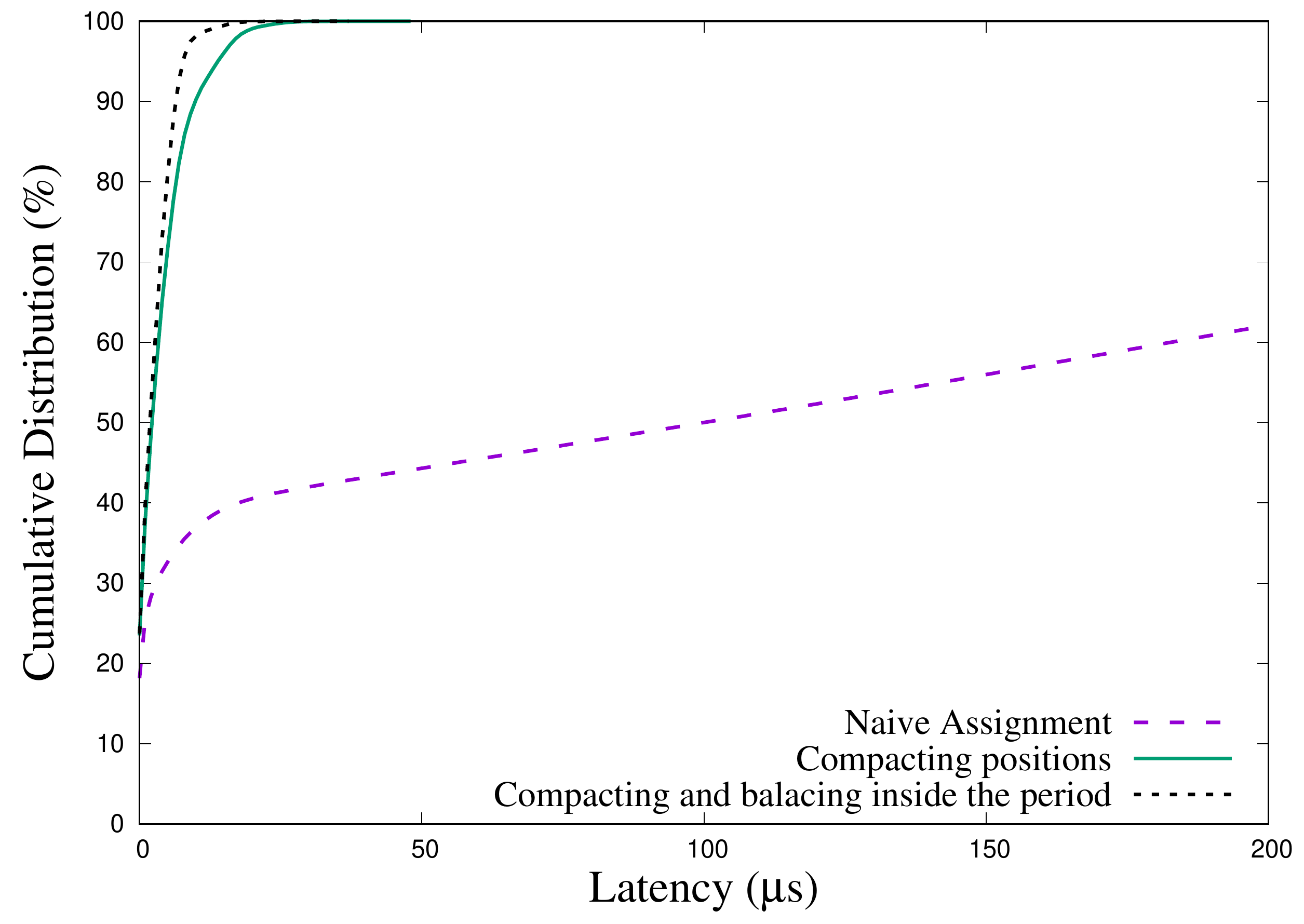}
     \captionof{figure}{BE latencies of compacting positions and balancing inside the period for $12$ antennas.}   \label{fig:algocmp}
     \end{center}

  \begin{center}
     \includegraphics[scale=0.3]{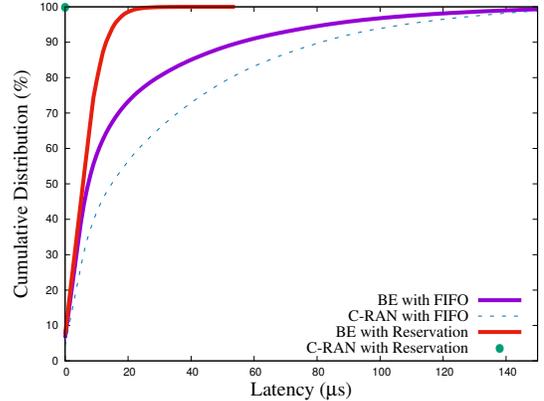}
     \captionof{figure}{FIFO buffer compared to the best method with reservation for $12$ antennas.} \label{fig:optimres}
     \label{fig:periodonly}
     \end{center}

\subsection{Building valid assignment with some C-RAN latency}
\label{sec:maxant}

The previous approach limits the number of antennas supported by the ring when $P-RS \mod ET \neq 0$, which is the case with N-GREEN parameters. The method we present in this section enables us to support more antennas and improves the latency of BE traffic (it reserves less free containers) by \emph{allowing the data from an RRH to use two positions}.
It is at the cost of a slightly worse latency for C-RAN traffic and it also requires in practice to implement some buffering for the C-RAN packets. 

In order to support as much antennas as possible on the ring, we use \emph{all} containers in a given position, improving on the compacting position heuristic. 

\begin{prop}\label{prop:filling}
 There is a valid assignment for $k$ antennas when $k \leq \lfloor \frac{P- RS}{ET} \times \frac{F}{2}\rfloor$.
\end{prop}
\begin{proof}
 We consider the RRHs in the order of the ring.
 Let $l = \lfloor \frac{P- RS}{ET}\rfloor$, then we set the offsets of the first $l$ RRHs as in Prop.~\ref{prop:assign}. These RRHs are at position zero and the $(l+1)$th RRH first emits at position zero, with offset $m_{l+1} = l*ET + \omega(u_0,u_{l+1})$. 
 
 The $(l+1)$th RRH emits up to time $P - \omega(u_{l+1},u_{0})$ at position zero, so that there is no conflict with RRH $0$ during the next period.
 Hence, it has used the position zero during $x = P - \omega(u_{l+1},u_{0}) - l*ET - \omega(u_0,u_{l+1}) = P - l*ET - RS$. From time $P - \omega(u_{l+1},u_{0}) + 2$, the $(l+1)$th RRH emits at position $2$ and during a time $ET - x$. Then the next RRH in the order is assigned to position $2$, and begins to emit at time $P - \omega(u_{l+1},u_{0}) + ET -x$ instead of zero. The rest of the assignment is built in the same way filling completely all first positions, until there are no more RRH.  
\end{proof}

\begin{center}
    \includegraphics[scale=0.7]{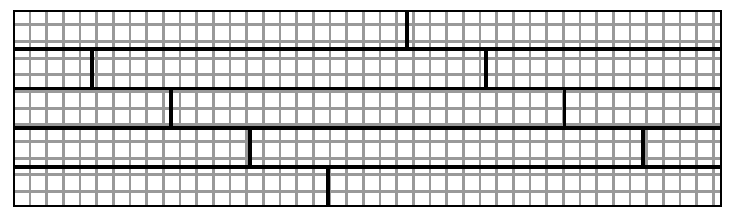}
     \captionof{figure}{Valid assignment for $9$ antennas and the N-GREEN parameters.}   \label{fig:split}
     \end{center}

Fig.~\ref{fig:split} illustrates the construction of Prop.~\ref{prop:filling} for the N-GREEN parameters. The loss due to reservation is exactly $RS$ containers by used positions. Hence, it is possible to support $9$ antennas (but no BE traffic in this extreme case),
rather than $5$ with the method of Sec.~\ref{sec:zerolatency}.

 We call this new reservation algorithm \textbf{saturating positions} since it improves on compacting positions of the previous subsection. Moreover, there are no free slots in used positions, hence the idea of balancing into the period is not relevant. The only possible optimisation would be to balance the used positions, but it is not worth it since it adds additional latency for the RRHs using two different positions.

  \begin{center}
  \includegraphics[scale=0.3]{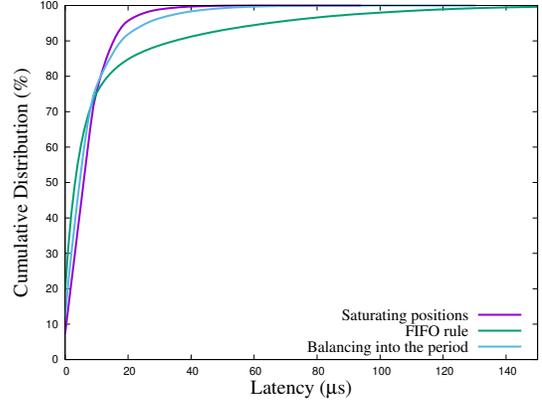}
     \captionof{figure}{Latencies of saturating positions, balancing into the period and FIFO rule for 5 antennas.}   \label{fig:splitres}
     \end{center}

%  
%   Since the N-GREEN parameters only allows to do balancing into the period of the previous section, we want to compare the impact on the best effort traffic of fulfilling positions against balancing into the period and FIFO rule. 

Fig.~\ref{fig:splitres} represents the cumulative distribution of the latency of BE traffic for the FIFO rule, saturating position, and balancing into the period using the N-GREEN parameters. Saturating positions reduces the BE traffic latency more than balancing into the period. This is easily explained by its lesser use of reservation. It is at the cost of a maximal latency of $2$ $\mu$s for C-RAN traffic, so the designer can chose to use any of the two algorithms, according to what latency must be guaranteed to C-RAN or BE traffic.

  \section*{Conclusion}
  
  As a conclusion, we want to stress the fact that to deal with a deterministic  dataflow as C-RAN, we must use a deterministic policy rather than a classical stochastic one.
  By using a simple practical SDN scheme, which requires only to set the emission timing of the RRHs and to allow reservation on the optical ring, we remove all logical latencies. It also improves significantly the latencies of the BE traffic by spreading the load of the C-RAN traffic uniformly over the period.  We are currently working on a prototype implementing this method on the NGREEN ring.
  We also plan to deal with the case of several data-centers containing the BBUs instead of one. The algorithmic methods to find good emission timings in this generalization are more complicated and similar to what was proposed in~\cite{dominique2018deterministic}, but while the load due to the C-RAN traffic is not too high it should work very well. 
  The results obtained show that te N-GREEN network architecture has a high potential in term of minimization of the end-to-end latency, in a multi-QoS environment. This study complete several studies demonstrating that the broadcast and select mechanism is extremely powerful to lead to deterministic networks, since it minimize the latency to its minimum feasible.
  
      \bibliographystyle{splncs04}
\bibliography{src}

\end{document}